\documentclass[10pt, letterpaper, twocolumn, conference]{IEEEtran}

\usepackage{graphicx}
\usepackage{subfigure}
\usepackage{amsmath}
\usepackage{amssymb}
\usepackage{cite}
\usepackage{cuted}
\usepackage{stfloats}
\usepackage{color}
\usepackage{flushend}
\usepackage{epstopdf}
\IEEEoverridecommandlockouts
\newtheorem{theorem}{\bf Theorem}
\newtheorem{lemma}{\bf Lemma}

\newcommand{\m}[1]{\mathbf{#1}^m}
\newcommand{\lo}[1]{\log_2\left(#1\right)}

\newcommand{\co}[1]{\ensuremath{C}_{#1}}

\newcommand{\expect}[1]{\mathbb{E}\left[{#1}\right]}

\definecolor{DarkGreen2}{rgb}{0.00,0.6,0.08}

\title{Is Witsenhausen's counterexample a relevant toy?}

\author{Pulkit Grover and Anant Sahai\\ Department of EECS, University of California at Berkeley, CA-94720, USA\\ \{pulkit, sahai\}@eecs.berkeley.edu}
\begin{document}\maketitle

\begin{abstract}
This paper answers a question raised by Doyle on the relevance of the Witsenhausen counterexample as a toy decentralized control problem. The question has two sides, the first of which focuses on the lack of an external channel in the counterexample. Using existing results, we argue that the core difficulty in the counterexample is retained even in the presence of such a channel. The second side questions the LQG formulation of the counterexample. We consider alternative formulations and show that the understanding developed for the LQG case guides the investigation for these other cases as well. Specifically, we consider 1) a variation on the original counterexample with general, but bounded, noise distributions, and 2) an adversarial extension with bounded disturbance and quadratic costs. For each of these formulations, we show that quantization-based nonlinear strategies outperform linear strategies by an arbitrarily large factor. Further, these nonlinear strategies also perform within a constant factor of the optimal, uniformly over all possible parameter choices (for fixed noise distributions in the Bayesian case).  

Fortuitously, the assumption of bounded noise results in a significant simplification of proofs as compared to those for the LQG formulation. Therefore, the results in this paper are also of pedagogical interest.
\end{abstract}
\section{Introduction}
Recently, we provided the first \textit{provably approximately optimal} solution to the Witsenhausen counterexample and its vector extensions~\cite{WitsenhausenJournal,FiniteLengthsWitsenhausen}. The solutions are obtained using techniques from information theory that help us understand the \textit{implicit communication} between the controllers: the ability of one controller to `talk' to the other by making changes to the state of the system. 
The counterexample was discussed quite a bit in the symposium on `Paths ahead in the science of information and decision systems' held in November 2009 at MIT LIDS in honor of Prof. Sanjoy Mitter. The ensuing discussions led Prof. John Doyle to question the relevance of the counterexample as a toy problem in decentralized control. The goal of this paper is to convince the reader of that relevance. 

It is hard to define what constitutes a useful and relevant toy problem. In order to obtain a better understanding of what such a problem could be, it is useful to look at the following toy problem from the neighboring field of  information theory: communicating a source across a power-constrained AWGN channel to minimize the average quadratic distortion in reconstructing the source. The problem is a toy because it caricatures the real world in three ways: communication problems today are never just point-to-point links, noise is rarely Gaussian, and a quadratic distortion cost is not the perceptually `correct' cost criterion for most sources~\cite{GrayPerceptual}. Even though these  assumptions make it a toy problem, it is a useful toy: it distills the problem of transmitting a source across a channel --- an aspect that is inherent in all practical problems --- in a minimalist fashion. A solution to this AWGN problem provided the foundation for system architectures (e.g. separation of source and channel coding) and coding techniques for larger communication problems (e.g. see~\cite{Tse2005}) including those with multiple transmitters and receivers, multiple antennas, non-Gaussian noise, etc. 

With this understanding, is Witsenhausen's counterexample a relevant toy? Similar to the point-to-point communication problem, the counterexample distills the possibility of implicit communication that appears to be ubiquitous in decentralized control systems. Why, then, may it not be relevant? Doyle's first argument rests on the work of Rotkowitz and Lall~\cite{RotkowitzLall}, which shows that with extremely fast, infinite-capacity, and perfectly reliable external channels, the optimal controllers are linear not just for the Witsenhausen counterexample (which is a simple observation), but for more general problems as well. Given that using an external channel is often a valid engineering option in decentralized control problems, Doyle argued that Witsenhausen's counterexample may be artificially hard because it does not allow the controllers to talk over an external channel, and instead forces the controllers to talk implicitly through the plant. The toy may be irrelevant: the architectural freedom of installing an external channel seemingly obviates any need for implicit communication. 

In practice, however, an external channel \textit{never} has infinite capacity or perfect reliability, which is what motivates a growing body of the control theory literature (for example~\cite{OurMainLQGPaper,MartinsBodeIntegral}) that addresses the issue of control over noisy and finite-capacity communication channels. In the presence of an imperfect external channel connecting the two controllers in Witsenhausen's counterexample, Martins~\cite{MartinsSideInfo} shows that while finding optimal solutions continues to be hard, one can design  signaling-based nonlinear strategies guided by those developed for the original counterexample. Martins also shows that in some cases, nonlinear strategies that do not even use the external channel can outperform linear strategies\footnote{
A similar problem is considered by Shoarinejad et al in~\cite{shoarinejad}, where noisy side information of the source is available at the receiver. Since the channel in formulation of~\cite{shoarinejad} is even more constrained than that in~\cite{MartinsSideInfo}, and nonlinear strategies outperform linear even without using the external channel for Martins's problem, they outperform linear for Shoarinejad's  problem as well.}. Provisioning for a very high SNR external channel, which has its own installation and operating costs, may therefore be unnecessary as long as nonlinear control techniques are used. In a companion paper~\cite{Allerton10Paper}, we consider this problem in greater detail and show that signaling-based nonlinear strategies can outperform linear ones by an arbitrarily large factor for any chosen finite-capacity external channel. 
We also derive approximately-optimal strategies which \textit{do} make use of the external channel\footnote{As is suggested by what David Tse calls the ``deterministic perspective" (along the lines of~\cite{DeterministicModel,SalmanThesis,DeterministicApproach}),  linear strategies do not make good use of the external channel because they only communicate the ``most significant bits'' --- which can be estimated reliably at the second controller anyway. So if the uncertainty in the initial state is large, the external channel is only of limited help and there remains a substantial advantage in having the controllers also talk through the plant.}, but even these results build on an understanding of the original counterexample, justifying its relevance as a toy problem. 

Doyle's second argument is about the relevance of the LQG framework in Witsenhausen's counterexample. Linearity is fine, but do we believe that primitive random variables are  Gaussian? Or that the designer is wedded to quadratic costs? The answer is no! Primitive random variables are almost never Gaussian, and the cost function is chosen more freely by the designer --- the quadratic case is only one amongst many possible formalizations of the intuition that the cost increases at an increasing rate. As suggested by Doyle, of interest here is the work of Rotkowitz~\cite{rotkowitz}. Rotkowitz shows that for the adversarial $L_2$-induced norm, as opposed to the original expected quadratic cost in Witsenhausen's formulation, linear control laws \textit{are} optimal and easy to find. At the same time, noise and initial state realizations can be completely arbitrary. Doyle's implicit argument, based on Rotkowitz's observation, is that because there is nothing sacred about the choice of a norm, viewed through the lens of a different (\textit{i.e.} induced) norm (and with fewer assumptions), Witsenhausen's problem does not require implicit communication!\footnote{It does not appear that this was Rotkowitz's original motivation. He was motivated because the idea was surprising enough that no one believed him~\cite{RotkowitzSlides}.} Indeed, with an induced norm, the problem seems no more intriguing than other team-theoretic problems with two controllers. 


The rest of this paper addresses this second argument. The induced-norm takes a frequentist's  approach and further assumes that \textit{nothing} is known about the state and noise values --- they can be completely arbitrary. The control strategy is therefore paranoid, and budgets for all possible values of state and noise, fearing for the worst. In the cost function, this is reflected as a maximization over the state and noise values. Because no assumptions are made on how large the noise and state values can be, maximization of an unnormalized quadratic cost would diverge to infinity for any control scheme. To prevent this, the maximization is performed over a quadratic function of state and noise that is \textit{normalized} with the size (a quadratic sum) of state and noise realizations. This  ``gain-perspective'' is commonly adopted in understanding input-output stability~\cite[Pg. 430]{khalil} of nonlinear systems. It characterizes how the norm of a signal changes as it passes through a system.


In practical engineering contexts, however, one often knows the ``typical'' values of state perturbations and noise realizations. The normalization in the gain perspective then does not reflect the actual costs incurred by the system. For instance, when the state perturbations and observation noises are small, the state estimates are more reliable, and therefore the control costs are often smaller. While a plain quadratic cost criterion reflects these smaller costs, the gain-perspective of induced-norm approach does not. 

In order to demonstrate our point, we look at the counterexample from both Bayesian and frequentist perspectives. To model the knowledge of ``typical'' values of primitive random variables, we assume merely that the \textit{noise} is bounded, and this bound is known. Our Bayesian model (Section~\ref{sec:bayesian}) is inspired from uniformly distributed noise. It considers an average quadratic cost assuming further that the distribution of the initial state is Gaussian, but departs from the LQG model in that the distribution of noise is bounded and known. Our frequentist model (Section~\ref{sec:frequentist}) goes a step further and considers a worst-case unnormalized quadratic cost assuming there is no prior distribution on the state and the noise. Yet, for both of these formulations, implicit communication can not be ignored. Quantization-based implicit-communication strategies can outperform linear strategies by an arbitrarily large factor\footnote{These results are based on similar results by Mitter and Sahai~\cite{AreaExamPaper} for the original counterexample.}, and these strategies also attain within a constant factor of the optimal cost. In the Bayesian case, the constant factor is reasonably small for uniform noise, as it was for the Gaussian case in~\cite{WitsenhausenJournal,FiniteLengthsWitsenhausen}, but it can be large for other distributions. When it is large, improved implicit-communication strategies will be needed in order to attain within a small constant factor.



Fortuitously, the proofs for bounded noise formulations considered in this paper are substantially simpler than those for the LQG formulation --- a finite-length analysis in the style of~\cite{FiniteLengthsWitsenhausen} is not needed to show approximate optimality\footnote{Even though a finite-length analysis is needed to obtain tighter bounds on the associated constant factors.}.

\vspace{-0.05in}

\section{Notation and problem statement}
\label{sec:probstat}
\begin{figure}[htb]
\begin{center}
  \includegraphics[scale=0.4]{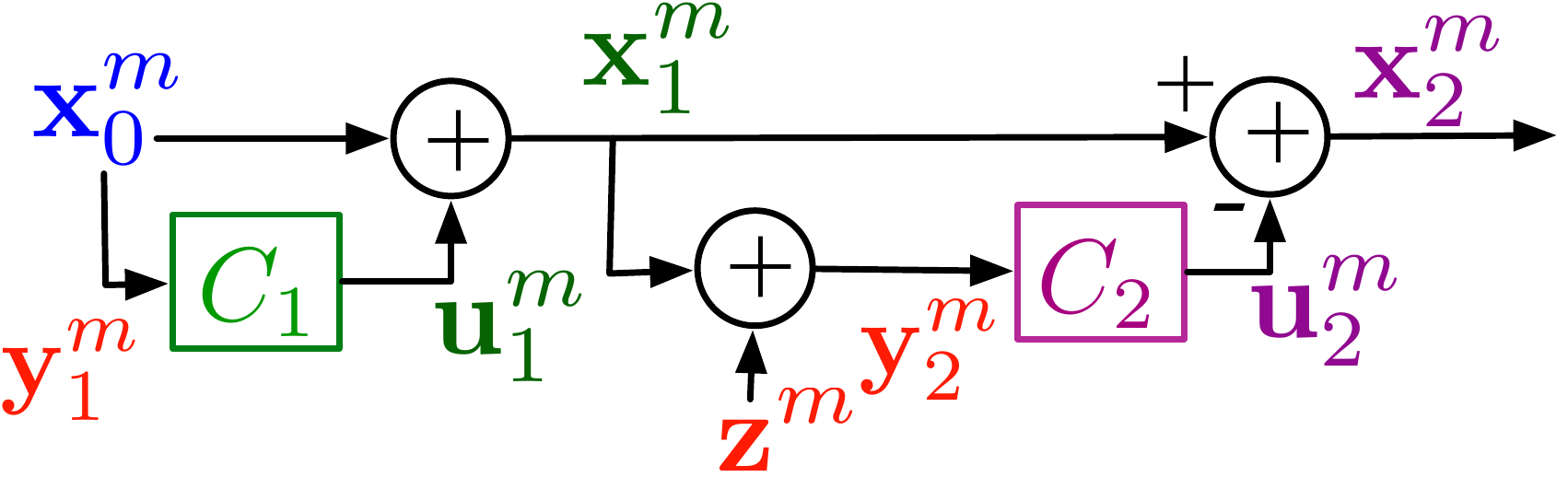}
\caption{Block-diagram for the vector Witsenhausen counterexample~\cite{WitsenhausenJournal}. }
\label{fig:infotheory}
\end{center}
\end{figure}

\vspace{-0.15in}

Vectors are denoted in bold, with the superscript to denote their length (e.g. $\m{x}$ is a vector of length $m$). Upper case is used for random variables or random vectors (except when denoting power $P$), while lower case symbols represent their realizations. Hats $(\,\widehat{\cdot{}}\,)$ on the top of random variables denote the estimates of the random variables. The block-diagram for the formulations considered in this paper is shown in Fig.~\ref{fig:infotheory}. 

A control strategy is denoted by $\gamma=(\gamma_1,\gamma_2)$, where $\gamma_i$ is the function that maps the observation $\m{y}_i$ at $\co{i}$ to the control input $\m{u}_i$. The observations are given by $\m{y}_1= \m{x}_0$ and $\m{y}_2=\m{x}_1 + \m{z}$, where $\m{z}$ is the disturbance, or the noise at the input of the second controller. For the first two formulations, the total cost is a quadratic function of the state and the input given by:
\begin{equation}
\label{eq:cost}
J^{(\gamma)}(\m{x}_0,\m{z}) = \frac{1}{m}k^2\|\m{u}_1\|^2+\frac{1}{m}\|\m{x}_2\|^2,
\end{equation}
where $\m{u}_1=\gamma_1(\m{x}_0)$, $\m{x}_2=\m{x}_0+\gamma_1(\m{x}_0)-\m{u}_2$ where $\m{u}_2 = \gamma_2(\m{x}_0+\gamma_1(\m{x}_0) + \m{z})$.  The cost expression includes a division by the vector-length $m$ to allow for natural comparisons between different vector-lengths. 

We now provide the two problem formulations that are addressed in this paper.

\vspace{-0.05in}

\subsection{Bayesian approach: a stochastic formulation}
\label{sec:stocstat}
The initial state $\m{X}_0$ is Gaussian, distributed $\mathcal{N}(0,\sigma_0^2\mathbb{I}_m)$, where $\mathbb{I}_m$ is the identity matrix of size $m\times m$. The observation noise $\m{Z}$ is distributed \textit{iid} according to distribution $f_Z(z)$ with finite differential entropy $h(Z)$, finite variance $\sigma_z^2$, and bounded support contained in $(-a,a)$. Without loss of generality, we assume that $\sigma_z^2=1$. For example, for a uniformly distributed $Z$, $\sigma_z^2=1$ for $a=\sqrt{3}$. 

The control objective is to minimize the expected quadratic cost $\overline{J}^{(\gamma)}$, 

\vspace{-0.2in}

\begin{equation}
\overline{J}^{(\gamma)} = \expect{J^{(\gamma)}} = \frac{1}{m}k^2\expect{\|\m{U}_1\|^2}+\frac{1}{m}\expect{\|\m{X}_2\|^2},
\end{equation}
over the choice of $\gamma$. The cost is averaged over the  random realizations of $\m{X}_0$ and $\m{Z}$. We use the variable $P:=\frac{1}{m}\expect{\|\m{U}_1\|^2}$ to denote the \textit{power} of the input $\m{u}_1$, and minimum mean-square error $MMSE = \frac{1}{m}\expect{\|\m{X}_2\|^2}=\frac{1}{m}\expect{\|\m{X}_1-\m{U}_2\|^2}$ to denote the second stage cost.

\subsection{Frequentist approach: an adversarial formulation with quadratic cost}
\label{sec:quadstat}
The block-diagram is the same as that for the stochastic problem. The total cost is still the same function given by~\eqref{eq:cost}, however, the cost for a strategy $\gamma$ is given by the maximum cost under the constraint\footnote{The bound of $\sqrt{3}$ is so chosen because it simplifies the derivations of upper and lower bounds. } that $|z_i|< \sqrt{3}$ for all $i$. That is,
\begin{equation}
J^{(\gamma)}_{freq} = \sup_{\m{x}_0,\|\m{z}\|_{\infty}< \sqrt{3}} J^{(\gamma)}(\m{x}_0,\m{z}).
\end{equation}

\vspace{-0.1in}

\section{Stochastic models for state and noise}
\label{sec:bayesian}

\vspace{-0.05in}

\subsection{Upper bound on costs}

\vspace{-0.05in}

\begin{theorem}
\label{thm:upper}
An upper bound on the optimal average costs, $\overline{J}_{opt}$, for the stochastic problem of Section~\ref{sec:stocstat} is given by
\begin{equation}
\overline{J}_{opt}\leq \min\left\{ k^2 a^2, \frac{\sigma_0^2}{\sigma_0^2+1}, k^2\sigma_0^2  \right\}.
\end{equation}
\end{theorem}
\begin{proof}
We consider the following three strategies 1) a scalar quantization strategy that quantizes the entire real line using uniform quantization-bins of size $2a$ in each dimension, 2) the zero-input strategy, followed by LLSE estimation at the second controller, and 3) the zero-forcing strategy. For a given $(k,\sigma_0)$-pair, the strategy with minimum cost is chosen. 

For the quantization strategy, the input forces the state to the nearest quantization point. The magnitude of the input is therefore bounded by $a$. Since the bins are disjoint, there are never any errors at the second controller (because the noise is smaller than $a$). The total cost is therefore upper bounded by $k^2a^2$. For zero-input strategy with Linear Least-Square Estimation (LLSE), the cost is the same as that in the Gaussian case zero-input strategy of~\cite{WitsenhausenJournal} of $\frac{\sigma_0^2}{\sigma_0^2+1}$ (because MMSE and LLSE operations are the same in the Gaussian formulation, and LLSE error depends on the distribution only through the variance of the random variable). For zero-forcing, the input is forced to zero, and thus the cost is $k^2\sigma_0^2$. This completes the proof. \end{proof}

\vspace{-0.05in}

\subsection{A lower bound on the costs}
\begin{theorem}
\label{thm:lowerbound}
A lower bound on the costs for the stochastic problem of Section~\ref{sec:stocstat} with observation noise $Z$ of variance $1$ and differential entropy $h(Z)$ is given by
\begin{equation}
\overline{J}_{opt}\geq \inf_{P\geq 0} k^2 P  + \left( \left(  \sqrt{\kappa(P)} -\sqrt{P}   \right)^+ \right)^2,
\end{equation}
where
\begin{equation}
\kappa (P) = \frac{\sigma_0^2 2^{2h(Z)}} {2\pi e \left( (\sigma_0 + \sqrt{P})^2+ 1 \right)}.
\end{equation}
\end{theorem}
\begin{proof}
The proof follows the lines of the proof of Theorem 3 in~\cite{WitsenhausenJournal}. For a fixed $P:=\frac{1}{m}\expect{\|\m{U}_1\|^2}$, we first obtain a lower bound on the $MMSE$. We need the following lemma~\cite[Lemma 3]{WitsenhausenJournal}.
\begin{lemma}
\label{lem:triangle}
For any three random vectors $A$, $B$ and $C$,
\begin{equation*}
\sqrt{\expect{\|B-C\|^2}}\geq \sqrt{\expect{\|A-C\|^2}} - \sqrt{\expect{\|A-B\|^2}}.
\end{equation*}
\end{lemma}
\begin{proof}
See~\cite{WitsenhausenJournal}.
\end{proof}
Substituting $\m{X}_0$ for $A$, $\m{X}_1$ for $B$, and $\m{U}_2$ for $C$ in Lemma~\ref{lem:triangle},
\begin{eqnarray}
\label{eq:sqrteqn}
\nonumber  && \sqrt{\expect{\|\m{X}_1-\m{U}_2\|^2}}\\
&& \geq  \sqrt{\expect{\|\m{X}_0-\m{U}_2\|^2}}-\sqrt{\expect{\|\m{X}_0-\m{X}_1\|^2}}.
\end{eqnarray}
We wish to lower bound  $\expect{\|\m{X}_1-\m{U}_2\|}$. The second term on the RHS is smaller than $\sqrt{mP}$. Therefore, it suffices to lower bound the first term on the RHS of~\eqref{eq:sqrteqn}. If we interpret $\m{U}_2$ as an estimate for $\m{X}_0$, this term represents the MMSE in reconstruction of $\m{X}_0$ across the $X_1-Y_2$ channel. 
\begin{lemma}
\label{lem:mutinfo}
The mutual information $I(\m{X}_1;\m{Y}_2)$ is bounded as follows

\vspace{-0.25in}

\begin{equation}
\frac{1}{m}I(\m{X}_1;\m{Y}_2)\leq \frac{1}{2}\lo{\frac{2\pi e \left( (\sigma_0+\sqrt{P})^2+1\right)}{2^{2h(Z)}}}.
\end{equation}
\end{lemma}
\begin{proof}
See Appendix~\ref{app:mutinfo}.
\end{proof}
We can now obtain a lower bound on the MMSE in reconstructing $\m{X}_0$ as follows: $\m{X}_0$ is a Gaussian source that is reconstructed across a channel of mutual information (and hence also the capacity) upper bounded by the expression in~\eqref{eq:mutinfo}. The MMSE in reconstructing $\m{X}_0$ is therefore lower bounded by $mD_{\sigma_0^2}(C_{X_1-Y_2})$ where $D_{\sigma_0^2}(R):=\sigma_0^2 2^{-2R}$ is the distortion-rate function~\cite[Ch. 13]{CoverThomas} of a Gaussian source, and $C_{X_1-Y_2}$ is the capacity across the $X_1-Y_2$ channel.

Thus, the MMSE in reconstructing $\m{X}_0$ is lower bounded by
\begin{eqnarray}
\label{eq:lasteqn}
\nonumber \frac{1}{m}\expect{\|\m{X}_0-\m{U}_2\|^2} &\geq & D_{\sigma_0^2}(C_{X_1-Y_2})\\
&\geq &\frac{\sigma_0^2 2^{2h(Z)}} {2\pi e\left( (\sigma_0+\sqrt{P})^2+1\right)}.
\end{eqnarray}

\vspace{-0.05in}

A lower bound on the $MMSE$ follows from~\eqref{eq:sqrteqn} and~\eqref{eq:lasteqn}. The theorem follows from the minimizing the sum of $k^2P$ and $MMSE$ over non-negative values of $P$.
\end{proof}
Observe that the proof does not make use of the bounded nature of the noise. The theorem is thus  applicable to Gaussian noise as well, and is therefore a generalization of the lower bound in~\cite{WitsenhausenJournal}.

\vspace{-0.05in}

\subsection{Quantization-based strategies are approximately optimal}
We now show that the upper bound in Theorem~\ref{thm:upper} is within a constant factor of the lower bound in Theorem~\ref{thm:lowerbound}.
\begin{theorem}
For the problem as stated in Section~\ref{sec:stocstat},

\vspace{-0.2in}

\begin{eqnarray*}
&&\inf_{P\geq 0} k^2 P  + \left( \left(  \sqrt{\kappa(P)} -\sqrt{P}   \right)^+ \right)^2 \leq \overline{J}_{opt}\\
&& \leq \mu\left(\inf_{P\geq 0} k^2 P  + \left( \left(  \sqrt{\kappa(P)} -\sqrt{P}   \right)^+ \right)^2\right),
\end{eqnarray*}
where $\mu \leq \frac{200a^2}{2^{2h(Z)}}$, and the upper bound is achieved by quantization-based strategies, complemented by linear strategies. For example, for $Z\sim \mathbb{U}(-\sqrt{3},\sqrt{3})$, the uniform distribution of variance $1$, $\mu\leq 50$.
\end{theorem}
\begin{proof}
The proof is along the lines of proof of Theorem 1 of~\cite{WitsenhausenJournal}. We use $P^*$ to denote the optimizing value of $P$ in the lower bound. We consider two cases:

\textit{Case 1:}  $\sigma_0^2 <1$.

If $P^*>\frac{\sigma_0^2 2^{2h(Z)}}{200}$, using zero-forcing strategy, we have an upper bound of $k^2\sigma_0^2$. The lower bound is larger than $k^2P^*$ which in this case is larger than $k^2\frac{\sigma_0^2 2^{2h(Z)}}{200}$. The ratio is thus smaller than $\frac{200}{2^{2h(Z)}}$. 

If $P^*\leq \frac{\sigma_0^2 2^{2h(Z)}}{200}$, 
\begin{eqnarray*}
\kappa(P) &= &\frac{\sigma_0^2 2^{2h(Z)}}{2\pi e  \left( (\sigma_0+\sqrt{P^*})^2  + 1 \right)}\\
&\overset{\sigma_0^2\leq 1, P^*\leq \frac{\sigma_0^2 2^{2h(Z)}}{200}}{\geq} & \frac{\sigma_0^2 2^{2h(Z)}}{2\pi e  \left( (1+\sqrt{\frac{ 2^{2h(Z)}}{200}})^2  + 1 \right)}\\
&\overset{(a)}{\geq} & \frac{\sigma_0^2 2^{2h(Z)}}{2\pi e  \left( (1+\sqrt{\frac{ \pi e}{100}})^2  + 1 \right)} \\
&\approx & \frac{\sigma_0^22^{2h(Z)}}{41.95} > \frac{\sigma_0^22^{2h(Z)}}{42},
\end{eqnarray*}
where $(a)$ follows from the fact that $h(Z)\leq \frac{1}{2}\lo{2\pi e}$, the differential entropy for the $\mathcal{N}(0,1)$ random variable (Gaussian distribution maximizes the differential entropy for a given variance). Thus,
\begin{eqnarray*}
\left(\left(\kappa - \sqrt{P^*}\right)^+\right)^2&\geq & \sigma_0^2 2^{2h(Z)} \left(\frac{1}{\sqrt{42} } - \frac{1}{\sqrt{200}}\right)^2 \\
& \approx & \frac{\sigma_0^22^{2h(Z)}}{143.11}> \frac{\sigma_0^22^{2h(Z)}}{145},
\end{eqnarray*}
which is also a lower bound on the total cost. Using the zero-input upper bound of $\frac{\sigma_0^2}{\sigma_0^2+1}<\sigma_0^2$, the ratio in this case is upper bounded by $\max\left\{\frac{200}{2^{2h(Z)}},\frac{145}{2^{2h(Z)}}\right\}$.

\textit{Case 2:} $\sigma_0^2\geq 1$.

If $P^*> \frac{2^{2h(Z)}}{200}$, using the upper bound of $k^2 a^2$, the ratio of upper and lower bounds is smaller than $\frac{k^2 a^2}{k^2\frac{2^{2h(Z)}}{200} } = \frac{200 a^2}{2^{2h(Z)}}$. 

If $P^*\leq \frac{2^{2h(Z)}}{200}\leq \frac{2\pi e}{200}$ (again, because Gaussian distribution maximizes the differential entropy for given variance),
\begin{eqnarray*}
\kappa(P) &\overset{(a)}{=} &\frac{\sigma_0^2 2^{2h(Z)}}{2\pi e  \left( (\sigma_0+\sqrt{P^*})^2  + 1 \right)}\\
&\overset{(b)}{\geq} & \frac{2^{2h(Z)}}{2\pi e  \left( (1+\sqrt{P^*})^2  + 1 \right)}\\
&\geq & \frac{2^{2h(Z)}}{2\pi e  \left( (1+\frac{\pi e}{100})^2  + 1 \right)}\geq \frac{2^{2h(Z)}}{42},
\end{eqnarray*}
where $(b)$ holds because the expression in the RHS of $(a)$ is an increasing function of $\sigma_0$. Thus, the following lower bound holds for the MMSE error
\begin{eqnarray*}
MMSE\geq 2^{2h(Z)}\left(\frac{1}{\sqrt{42}}  -\frac{1}{\sqrt{200}}\right)^2\geq  \frac{2^{2h(Z)}}{145}.
\end{eqnarray*}
Using the zero-input upper bound, the ratio is smaller than $\frac{145}{2^{2h(Z)}}$. The ratio in this case is therefore smaller than $\max\left\{\frac{200a^2}{2^{2h(Z)}},\frac{145}{2^{2h(Z)}}\right\}\leq \frac{200a^2}{2^{2h(Z)}}$. The result now follows from the observation that $a\geq 1$, since the variance of $Z$ is $1$.
\end{proof}
Note that the result is not asymptotic --- the constant factor is uniform over all vector lengths, though it can be improved using lattice-based strategies of~\cite{FiniteLengthsWitsenhausen} for upper bound, and sphere-packing bounds~\cite{FiniteLengthsWitsenhausen} for lower bound.

\textit{Remark:} For the original counterexample, our results in~\cite{FiniteLengthsWitsenhausen} provide a constant factor that is uniform over the problem parameters $(k,\sigma_0^2)$. The constant factor of $\frac{200a^2}{2^{2h(Z)}}$ here depends on $a$ and $h(Z)$, and is therefore uniform over all $(k,\sigma_0^2)$ but only for a fixed noise distribution (and hence fixed $a$ and $h(Z)$). It blows up when the noise distribution has a long tail, or has a hugely negative differential entropy. In such cases, greater care is required in the design of implicit communication strategies. For instance, if the distribution is long-tailed, the quantization points need not be separated by $a$, but they can instead be separated by a distance sufficiently large so that the probability of mistaking one quantization point for another at the second controller is low. This insight is used in~\cite{FiniteLengthsWitsenhausen} to obtain strategies for the Gaussian case. 


\subsection{Quantization-based strategies outperform linear strategies by an unbounded factor}
\label{sec:arbitrary}
Consider the scalar case. A linear constraint on the second controller forces it to perform an LLSE estimation on the output $Y_2$ in order to estimate $X_1$. The first controller, also linear, uses an input $U_1=\alpha X_0$. The resulting state $X_1=(1+\alpha)X_0$ has variance $\widetilde{\sigma}_0^2 = \sigma_0^2(1+\alpha)^2$. The mean-squared estimation error is, therefore, $\frac{\widetilde{\sigma}_0^2}{\widetilde{\sigma}_0^2+1}$. Since this is an increasing function of $\widetilde{\sigma}_0^2$, the optimizing $\alpha$ is negative. Since $\alpha^2\sigma_0^2=P$, $\alpha\sigma_0=-\sqrt{P}$. The total cost for the optimal linear strategy is 
\begin{equation}
\overline{J}_{lin} = k^2 P + \frac{\left(\left(\sigma_0-\sqrt{P}\right)^+\right)^2}{\left(\left(\sigma_0-\sqrt{P}\right)^+\right)^2+1}.
\end{equation} 
Clearly, this cost remains the same in the vector case as well. 

We now consider two cases. If $P<\frac{\sigma_0^2}{4}$, 
\begin{eqnarray*}
\overline{J}_{lin}  \geq  \frac{\left(\left(\sigma_0-\sqrt{P}\right)^+\right)^2}{\left(\left(\sigma_0-\sqrt{P}\right)^+\right)^2+1}
 \overset{(a)}{\geq}  \frac{\sigma_0^2}{\sigma_0^2+4},
\end{eqnarray*}
where $(a)$ follows from the fact that $P<\frac{\sigma_0^2}{4}$. In the limit of $\sigma_0^2\rightarrow\infty$ and $k\rightarrow 0$, this lower bound increases to $1$, whereas the quantization upper bound of $k^2a^2$ decreases to zero. 

Alternatively, if $P\geq\frac{\sigma_0^2}{4}$, 
\begin{eqnarray}
\overline{J}_{lin}\geq k^2P\geq \frac{k^2\sigma_0^2}{4}.
\end{eqnarray} 
Thus, the ratio of the costs attained by the optimal linear strategy and those attained by the quantization upper bound is larger than $\frac{k^2\frac{\sigma_0^2}{4}}{k^2a^2}=\frac{\sigma_0^2}{4a^2}$ which diverges to infinity as $k\rightarrow 0,\;\sigma_0^2\rightarrow\infty$. 

\vspace{-0.05in}

\section{Adversarial model for noise and state}
\label{sec:frequentist}
\begin{theorem}
The optimal cost $J_{opt,freq}$ for adversarially modeled initial state and (bounded) noise $Z\in (-\sqrt{3},\sqrt{3})$ with quadratic costs (as defined in Section~\ref{sec:quadstat}) is bounded as follows  

\vspace{-0.3in}

\begin{eqnarray*}
&&\inf_{P\geq 0} k^2 P  + \left( \left(   \sqrt{\frac{ 6} {\pi e }} -\sqrt{P}   \right)^+ \right)^2\leq J_{opt,freq}\\
&&\leq 2\pi e\left(\inf_{P\geq 0} k^2 P  + \left( \left(   \sqrt{\frac{ 6} {\pi e }} -\sqrt{P}   \right)^+ \right)^2\right),
\end{eqnarray*} 
where the upper bound is achieved using quantization-based strategies complemented by linear strategies. Further, in the regime of $k\rightarrow 0$, the ratio of the costs attained by the best linear strategy to that attained by appropriate quantization-based nonlinear strategies diverges to infinity.
\end{theorem}

\begin{proof}
\textit{Upper bound:} If $k^2\leq 1$, we use a uniform quantization strategy  with bin size $2\sqrt{3}$. Since the noise amplitude is smaller than $\sqrt{3}$, there are no errors at the second controller. The cost of this strategy is therefore $3k^2$, which is attained in the event when the initial state is exactly at the edge of one of the quantization bins. 

If $k^2>1$, we use the zero-input strategy --- the first controller inputs zero, and the second controller chooses $\m{U}_2=\m{Y}_2$ as the estimate of $\m{X}_1$. Since noise amplitude is bounded by $\sqrt{3}$, the normalized error for this strategy is bounded by $3$.

The upper bound is therefore given by $\min\{3k^2,3\}$. 

\textit{Lower bound:} Even though the noise is chosen adversarially (and deterministically), we first assume that the noise behaves as a random variable with distribution $\mathbb{U}(-\sqrt{3},\sqrt{3})$, and the initial state behaves as a Gaussian with variance $\sigma_0^2$ for some $\sigma_0^2>0$. We assume that the adversary declares this strategy in advance (which can only reduce the costs). From Theorem~\ref{thm:lowerbound}, if the first controller chooses an average power $P$, then the $MMSE$ at the second controller is lower bounded by

\vspace{-0.14in}

\begin{equation}
MMSE\geq \left( \left(  \sqrt{\kappa(P)} -\sqrt{P}   \right)^+ \right)^2.
\end{equation}
Since this lower bound holds for all $\sigma_0^2$, we let $\sigma_0^2\rightarrow\infty$, and obtain the following bound, \begin{eqnarray*}
&& MMSE\geq  \left(\left(\sqrt{\lim_{\sigma_0^2\rightarrow\infty}\kappa(P)} - \sqrt{P}\right)^+\right)^2\\
&=& \left( \left(  \sqrt{\frac{ 2^{2h(Z)}} {2\pi e }} -\sqrt{P}   \right)^+ \right)^2 \overset{(a)}{=}  \left( \left(  \sqrt{\frac{ 6} {\pi e }} -\sqrt{P}   \right)^+ \right)^2,
\end{eqnarray*}
where $(a)$ follows from the fact that $h(Z)=\lo{2\sqrt{3}}$ for $Z\sim \mathbb{U}(\sqrt{3},\sqrt{3})$.

A lower bound on the average costs (averaged over the initial state and noise realizations) for this problem is 
\begin{equation}
\label{eq:advlower}
\overline{J}_{opt}\geq \inf_{P\geq 0} k^2 P  + \left( \left(   \sqrt{\frac{ 6} {\pi e }} -\sqrt{P}   \right)^+ \right)^2.
\end{equation}
At this point, if the adversary is allowed to use randomized strategies, we already have a proof of the lower bound. But what if it is required to play deterministically? We invoke an argument inspired by the probabilistic method~\cite{ProbabilisticMethod} to address this requirement. For a fixed strategy $\gamma$, the lower bound in~\eqref{eq:advlower} holds on the cost $J^{(\gamma)}(\m{x}_0,\m{z})$  averaged over $\m{x}_0$ and $\m{z}$ for any choice of strategy $\gamma$. Thus, there exists a choice of realizations ${\m{x}_0}^{(\gamma)}$ and ${\m{z}}^{(\gamma)}$ such that the cost $J^{(\gamma)}(\m{x}_0,\m{z})$ is at least as large as what the lower bound says it must be on average if $\m{x}_0$ and $\m{z}$ were random. This cost is further lower bounded by the expression in~\eqref{eq:advlower}. This proves the lower bound.

\textit{Bounded ratios:} \textit{Case 1:}  $P^*<\frac{6}{4\pi e}$. In this case,
\begin{eqnarray*}
MMSE &\geq &\left(\sqrt{\frac{6}{\pi e}} -  \sqrt{P^*}\right)^2\\
&\geq & \left(\sqrt{\frac{6}{\pi e}} -  \frac{1}{2}\sqrt{\frac{6}{\pi e}}\right)^2=\frac{6}{4\pi e} = \frac{3}{2\pi e},
\end{eqnarray*}
which is also a lower bound on the cost. Thus the ratio of the zero-input upper bound (which is $3$) and this lower bound is smaller than $\frac{3\times 2\pi e}{3}=2\pi e$.

\textit{Case 2:} $P^* \geq \frac{6}{4\pi e}$.\\
In this case, the cost is no smaller than $k^2 P^*=k^2\frac{6}{4\pi e}$. Thus the ratio of quantization-based upper bound (which is $3k^2$) and this lower bound is smaller than $\frac{3k^2\times 2\pi e}{3k^2}=2\pi e$.

The ratio of the upper and lower bound is therefore always smaller than $2\pi e\approx 17.08$.

\textit{Nonlinear strategies can outperform linear by an arbitrary factor:} 
The costs attained by quantization-based strategies are bounded by $3k^2$, regardless of the adversary's strategy. This gives an upper bound on the cost of nonlinear strategies. For linear strategies, we want to provide a lower bound. As in the proof of constant factor optimality, assume that the noise behaves as $\mathbb{U}(-\sqrt{3},\sqrt{3})$, and the initial state behaves as $\mathcal{N}(0,\sigma_0^2)$. The average costs attained by any linear strategy are lower bounded by 
\begin{equation}
\label{eq:advlb}
\inf_P k^2P + \frac{\left(\left(\sigma_0 - \sqrt{P}\right)^+\right)^2}{\left(\left(\sigma_0 - \sqrt{P}\right)^+\right)^2+1}. 
\end{equation}
Again, using the probabilistic method, there exists a realization of initial state and noise that attains a lower bound no smaller than the average. This gives a lower bound of~\eqref{eq:advlb} on deterministic costs. The bounds on costs of linear and nonlinear strategies are the same as that in Section~\ref{sec:arbitrary}, with the substitution of $a$ by $\sqrt{3}$. The remaining proof is thus the same. 
\end{proof}

\vspace{-0.1in}

\section*{Acknowledgments}

\vspace{-0.05in}

We thank Gireeja Ranade for pointing out errors in an early version, and Kristen Woyach and the anonymous referees for their helpful comments. This research is supported by NSF grants CCF-0729122, CCF-0917212 and CNS-0932410.

\vspace{-0.05in}

\appendices{}
\section{An upper bound on the mutual information across the $X_1-Y_2$ channel}
\label{app:mutinfo}

\vspace{-0.2in}

\begin{eqnarray*}
&&I(\m{X}_1;\m{Y}_2)  =  h(\m{Y}_2) - h(\m{Y}_2|\m{X}_1)\\
&\leq & \sum_i h(Y_{2,i}) - h(\m{Y}_2|\m{X}_1)\\
& = & \sum_i \left(h(Y_{2,i}) - h(Y_{2,i}|X_{1,i})\right)= \sum_{i} I(X_{1,i};Y_{2,i})\\ 
& \overset{(a)}{=} & m I(X_1;Y_2|Q)= m \left(  h(Y_2|Q) - h(Y_2|X_1,Q)  \right)\\
& =& m \left(  h(Y_2|Q) - h(Y_2|X_1)  \right)\\
& \leq & m \left(  h(Y_2) - h(Y_2|X_1)  \right)\leq mI(X_1;Y_2),
\end{eqnarray*}
where in $(a)$, $X_1 = X_{1,i}$ if $Q=i$ (and $Y_2$ is defined similarly), and $Q$ is distributed uniformly on the discrete set $\{1,2,\ldots,m\}$. Because $Z$ is independent of $X_0$ and $U_1$, the variance of $Y_2=X_0+U_1+Z$ is maximized when $X_0$ (of power $\sigma_0^2$) and $U_1$ (of power $P$) are aligned, and it equals $(\sigma_0 + \sqrt{P})^2 + 1$. Thus, 
\begin{eqnarray}
\label{eq:mutinfo}
\nonumber I(X_1;Y_2) &=& h(Y_2) - h(Y_2|X_1)\\
\nonumber &=& h(Y_2) - h(Z)\\
\nonumber &\overset{(a)}{\leq} & \frac{1}{2}\lo{2\pi e \left((\sigma_0+\sqrt{P})^2 + 1\right)} - h(Z)\\
&=& \frac{1}{2}\lo{\frac{2\pi e \left( (\sigma_0+\sqrt{P})^2+1\right)}{2^{2h(Z)}}},
\end{eqnarray}
where $(a)$ follows from the observation that for given second moment of the random variable, the distribution that maximizes the differential entropy is Gaussian.

\bibliographystyle{IEEEtran}
\bibliography{IEEEabrv,MyMainBibliography}
\end{document}